\documentclass[11pt, a4paper]{article}

\makeatletter%

\usepackage{ifthen}



\makeatletter%

\makeatletter
\newcommand{\ifndef}[2]{\@ifundefined{#1}{#2}{}}
\makeatother

\newcommand{\mydef}[2]{\def#1{#2}}

\newcommand{\nospell}[1]{#1}  %

\makeatletter
\newcommand{\myusepackage}[2][]{\@ifpackageloaded{#2}{} %
{\ifthenelse{\equal{}{#1}} {\usepackage{#2}} {\usepackage[#1]{#2}} }}
\makeatother

\myusepackage{amssymb}
\myusepackage{amsmath}  %

\myusepackage{amsthm}  %

{}
\myusepackage[T1,OT2,T2A]{fontenc}
\DeclareTextSymbolDefault{\CYRYAT}{OT2}
\DeclareTextSymbolDefault{\cyryat}{OT2}
\DeclareTextSymbolDefault{\CYRFITA}{OT2}
\DeclareTextSymbolDefault{\cyrfita}{OT2}
\DeclareTextSymbolDefault{\CYRIZH}{OT2}
\DeclareTextSymbolDefault{\cyrizh}{OT2}
\myusepackage[utf8]{inputenc}
\let\f\relax

{}
\usepackage[greek,ukrainian,russian,english]{babel}
{}

\myusepackage{latexsym}
\myusepackage{amsfonts}
\myusepackage{units}    %
\newcommand{\dr}{\nicefrac}
\myusepackage{psfrag}   %
\myusepackage{url}   %
\myusepackage{hyperref}  %
\myusepackage{hyphenat}  %
\myusepackage{caption}   %
\myusepackage{microtype}   %
\myusepackage{marginnote}  %
\myusepackage{calc}  %
{}  %
\myusepackage{enumitem}  %
{}
\myusepackage[dvips]{graphicx}
\myusepackage[usenames,dvipsnames]{color}
{}
{}  %
{} %

{}     %
\newcommand{\dgCapDefinition}{Definition}
\newcommand{\dgCapDefinitions}{Definitions}
\newcommand{\dgCapPostulate}{Postulate}
\newcommand{\dgCapPostulates}{Postulates}
\newcommand{\dgCapExample}{Example}

\newcommand{\dgCapFact}{Fact}
\newcommand{\dgCapFacts}{Facts}
\newcommand{\dgCapQuestion}{Question}
\newcommand{\dgCapQuestions}{Questions}
\newcommand{\dgCapLemma}{Lemma}
\newcommand{\dgCapLemmas}{Lemmas}

\newcommand{\dgCapCorollary}{Corollary}
\newcommand{\dgCapCorollaries}{Corollaries}
\newcommand{\dgCapProposition}{Proposition}
\newcommand{\dgCapPropositions}{Propositions}
\newcommand{\dgCapClaim}{Claim}
\newcommand{\dgCapClaims}{Claims}
\newcommand{\dgCapTheorem}{Theorem}
\newcommand{\dgCapTheorems}{Theorems}
\newcommand{\dgCapProblem}{Problem}
\newcommand{\dgCapProblems}{Problems}
\newcommand{\dgCapRemark}{Remark}
\newcommand{\dgCapRemarks}{Remarks}
\newcommand{\dgCapConjecture}{Conjecture}
\newcommand{\dgCapConjectures}{Conjectures}
\newcommand{\dgCapResult}{Result}

\newcommand{\dgCapChapter}{Chapter}
\newcommand{\dgCapChapters}{Chapters}
\newcommand{\dgCapSection}{Section}
\newcommand{\dgCapSections}{Sections}
\newcommand{\dgCapSubsection}{Subsection}
\newcommand{\dgCapSubsections}{Subsections}
\newcommand{\dgCapFigure}{Figure}
\newcommand{\dgCapFigures}{Figures}
\newcommand{\dgCapEquation}{Equation}
\newcommand{\dgCapEquations}{Equations}
\newcommand{\dgCapExpression}{Expression}
\newcommand{\dgCapExpressions}{Expressions}
\newcommand{\dgCapInequality}{Inequality}
\newcommand{\dgCapInequalities}{Inequalities}
\newcommand{\dgProofOf}{\proofname\ of}
{}
\newcommand{\dgDefinition}{Definition}
\newcommand{\dgDefinitions}{Definitions}
\newcommand{\dgPostulate}{Postulate}
\newcommand{\dgPostulates}{Postulates}

\newcommand{\dgFact}{Fact}
\newcommand{\dgFacts}{Facts}
\newcommand{\dgQuestion}{Question}
\newcommand{\dgQuestions}{Questions}
\newcommand{\dgLemma}{Lemma}
\newcommand{\dgLemmas}{Lemmas}
\newcommand{\dgCorollary}{Corollary}
\newcommand{\dgCorollaries}{Corollaries}
\newcommand{\dgProposition}{Proposition}
\newcommand{\dgPropositions}{Propositions}
\newcommand{\dgClaim}{Claim}
\newcommand{\dgClaims}{Claims}
\newcommand{\dgTheorem}{Theorem}
\newcommand{\dgTheorems}{Theorems}
\newcommand{\dgProblem}{Problem}
\newcommand{\dgProblems}{Problems}
\newcommand{\dgRemark}{Remark}
\newcommand{\dgRemarks}{Remarks}
\newcommand{\dgConjecture}{Conjecture}
\newcommand{\dgConjectures}{Conjectures}

\newcommand{\dgChapter}{Chapter}
\newcommand{\dgChapters}{Chapters}
\newcommand{\dgSection}{Section}
\newcommand{\dgSections}{Sections}
\newcommand{\dgSubsection}{Subsection}
\newcommand{\dgSubsections}{Subsections}
\newcommand{\dgFigure}{Figure}
\newcommand{\dgFigures}{Figures}
\newcommand{\dgEquation}{Equation}
\newcommand{\dgEquations}{Equations}
\newcommand{\dgExpression}{Expression}
\newcommand{\dgExpressions}{Expressions}
\newcommand{\dgInequality}{Inequality}
\newcommand{\dgInequalities}{Inequalities}
{}
{}

{}
\ifndef{theorem}{\newtheorem{theorem}{\dgCapTheorem}}
\ifndef{theorem*}{\newtheorem*{theorem*}{\dgCapTheorem}}
{}

{}  %
{}  %
\ifndef{lemma}{\newtheorem{lemma}{\dgCapLemma}}
\ifndef{lemma*}{\newtheorem*{lemma*}{\dgCapLemma}}
\ifndef{corollary}{}
\ifndef{corollary*}{\newtheorem*{corollary*}{\dgCapCorollary}}
\ifndef{conjecture}{}
\ifndef{conjecture*}{\newtheorem*{conjecture*}{\dgCapConjecture}}
\ifndef{proposition}{}
\ifndef{proposition*}{\newtheorem*{proposition*}{\dgCapProposition}}
\ifndef{claim}{}
\ifndef{claim*}{\newtheorem*{claim*}{\dgCapClaim}}
\ifndef{result}{}
\ifndef{result*}{\newtheorem*{result*}{\dgCapResult}}
\ifndef{problem}{\newtheorem{problem}{\dgCapProblem}}
\ifndef{problem*}{\newtheorem*{problem*}{\dgCapProblem}}
{}  %
{}  %

\newtheoremstyle{mydefinition}  %
{\topsep}{\topsep}  %
{\slshape}  %
{}  %
{\bfseries}  %
{.}  %
{ }  %
{}  %

\newtheoremstyle{myremark}  %
{\topsep}{\topsep}  %
{\slshape}  %
{}  %
{\bfseries\itshape}  %
{.}  %
{ }  %
{\thmname{#1}\thmnumber{ \!#2}}  %

\newtheoremstyle{myexample}  %
{\topsep}{\topsep}  %
{\itshape}  %
{}  %
{\slshape}  %
{.}  %
{ }  %
{\ul{\thmname{#1}}}  %

\newtheoremstyle{myclaims}  %
{\topsep}{\topsep}  %
{\slshape}  %
{}  %
{\bfseries\slshape}  %
{.}  %
{ }  %
{\thmname{#1}\thmnumber{ \!#2}\ifthenelse{\equal{}{#3}}%
{}{\textnormal{ \!(#3)}}}  %

{} %
{} %
{\theoremstyle{myremark}

\newtheorem*{myremark*}{\dgCapRemark}
}
{} %
{\theoremstyle{mydefinition}
\ifndef{postulate}{}
\ifndef{postulate*}{\newtheorem*{postulate*}{\dgCapPostulate}}
\ifndef{definition}{}
\ifndef{definition*}{\newtheorem*{definition*}{\dgCapDefinition}}
}
{\theoremstyle{myexample}
\ifndef{example}{}
\ifndef{example*}{\newtheorem*{example*}{\dgCapExample}}
}
{\theoremstyle{myclaims}

\newtheorem*{my_claim*}{\dgCapClaim}
\ifndef{fact}{}
\ifndef{fact*}{\newtheorem*{fact*}{\dgCapFact}}
\ifndef{question}{}
\ifndef{question*}{\newtheorem*{question*}{\dgCapQuestion}}
}
{} %
{} %
{} %

{}
\newtheoremstyle{anystatementst}  %
{\topsep}{\topsep}  %
{\itshape}  %
{}  %
{\bfseries}  %
{.}  %
{ }  %
{#3}  %

{\theoremstyle{anystatementst} }

{}

\newcommand{\MyUniPat}{lsdfgkhjvrkjlhmisdlcjn}

\newcommand{\newident}[3][\MyUniPat]{\ifthenelse{\equal{\MyUniPat}{#1}}
{
\newcommand{#2}[1][]{\Ensuremath{\mathit{#3##1}}}
}
{\ifthenelse{\equal{}{#1}}
{
\newcommand{#2}[1][]{\Ensuremath{\mathit{#3}}}
}
{
\newcommand{#2}[1][\MyUniPat]{\ifthenelse{\equal{\MyUniPat}{##1}}%
{\Ensuremath{\mathit{#1}}}%
{\Ensuremath{\mathit{#3}}}}
}
}
}

\newcommand{\newidenT}[3][\MyUniPat]{\ifthenelse{\equal{\MyUniPat}{#1}}
{
\newcommand{#2}[1][\MyUniPat]{\ifthenelse{\equal{\MyUniPat}{##1}}%
{\il{#3}}%
{\Ensuremath{\mathit{#3##1}}}}
}
{
\newcommand{#2}[1][\MyUniPat]{\ifthenelse{\equal{\MyUniPat}{##1}}%
{\il{#1}}%
{\Ensuremath{\mathit{#3}}}}
}
}

\newcommand{\newmat}[3][\MyUniPat]{\ifthenelse{\equal{\MyUniPat}{#1}}%
{\newcommand{#2}[1][]{\Ensuremath{#3##1}}}%
{\newcommand{#2}[1][]{\Ensuremath{#3}}}%
}

\newcommand{\providemat}[3][\MyUniPat]{\ifthenelse{\equal{\MyUniPat}{#1}}
{\providecommand{#2}[1][]{\Ensuremath{#3##1}}}
{\providecommand{#2}[1][]{\Ensuremath{#3}}}  %
}

\newcommand{\newmatop}[3][\MyUniPat]{\ifthenelse{\equal{\MyUniPat}{#1}}
{
\mydef{#2}{\operatorname{#3}}
}
{
\newcommand{#2}[1][\MyUniPat]{\ifthenelse{\equal{\MyUniPat}{##1}}%
{\operatorname{#1}}%
{\operatorname{#3}}}
}
}

\newcommand{\newfunction}[2]{%
\newcommand{#1}[2][\MyUniPat]{\ifthenelse{\equal{\MyUniPat}{##1}}%
{\Ensuremath{#2\lf(##2\rt)}}%
{#2(##2)}}%
}

\makeatletter
\newcommand{\MyMakeTheoMacros}[3]{
\expandafter\newcommand\csname\expandafter\@gobble\string#2NostarNoname@DGaux\endcsname[2][]
{\ifthenelse{\equal{}{##1}}%
{\begin{#1}~##2 \end{#1}}%
{\begin{#1}\label{##1}~##2\end{#1}}%
}
\expandafter\newcommand\csname\expandafter\@gobble\string#2StarNoname@DGaux\endcsname[1]
{\begin{#1*}~##1 \end{#1*}}
\def#2{\expandafter\@ifstar%
\expandafter{\csname\expandafter\@gobble\string#2StarNoname@DGaux\endcsname}%
{\csname\expandafter\@gobble\string#2NostarNoname@DGaux\endcsname}%
}

\expandafter\newcommand\csname\expandafter\@gobble\string#2NostarName@DGaux\endcsname[3][]
{\ifthenelse{\equal{}{##1}}%
{\begin{#1}[\e{##2}]~##3 \end{#1}}%
{\begin{#1}[\e{##2}]\label{##1}~##3\end{#1}}%
}
\expandafter\newcommand\csname\expandafter\@gobble\string#2StarName@DGaux\endcsname[2]
{\begin{#1*}[\e{##1}]~##2 \end{#1*}}
\def#3{\expandafter\@ifstar%
\expandafter{\csname\expandafter\@gobble\string#2StarName@DGaux\endcsname}
{\csname\expandafter\@gobble\string#2NostarName@DGaux\endcsname}%
}
}
\makeatother

\makeatletter
\newtheorem*{rep@theorem}{\rep@title}
\newcommand{\newreptheorem}[2]{%
\newenvironment{rep#1}[1]{%
\def\rep@title{#2 \ref{##1}}%
\begin{rep@theorem}}%
{\end{rep@theorem}}}
\makeatother

\newcommand{\MyMakeDupTheoMacros}[7]{
\MyMakeTheoMacros{#1}{#2}{#3}
\newreptheorem{#1}{#6}
\newcommand{#4}[3]{
\newcommand{##2}{##3}
\begin{#1}\label{##1}~##2\end{#1}}
\newcommand{#5}[4]{
\newcommand{##2}{##4}
\begin{#1}{\e{##3}}\label{##1}~##2\end{#1}}
\newcommand{#7}[2]{\begin{rep#1}{##1}~##2 \end{rep#1}}
}

{}
\newcommand{\MyMakeRefMacros}[3]{\newcommand{#1}[2][]
{\ifthenelse{\equal{}{##1}}{#2~\ref{##2}}{#3~\ref{##1} and~\ref{##2}}}}

\newcommand{\MyMakeEqRefMacros}[3]{\newcommand{#1}[2][]
{\ifthenelse{\equal{}{##1}}{#2~\eqref{##2}}{#3~\eqref{##1} and~\eqref{##2}}}}
{}

{}  %
\newcommand{\bibentry}[8]{
{}\bibitem[\nospell{#8}]{#1} {\textup #3}.{}
\ifthenelse{\equal{}{#6}}
{\newblock \textrm{#4.} \newblock {\em #5}, #7....}
{\newblock \textrm{#4.} \newblock {\em #5, #6}, #7.}
}

{} %

\MyMakeTheoMacros{fact}{\fct}{\nfct}

\MyMakeRefMacros{\fctref}{\dgFact}{\dgFacts}
\MyMakeRefMacros{\Fctref}{\dgCapFact}{\dgCapFacts}

\MyMakeTheoMacros{question}{\quest}{\nquest}

\MyMakeRefMacros{\questref}{\dgQuestion}{\dgQuestions}
\MyMakeRefMacros{\Questref}{\dgCapQuestion}{\dgCapQuestions}

{} %
\MyMakeDupTheoMacros{lemma}
{\lem}{\nlem}{\lemdup}{\nlemdup}{\dgLemma}{\lemrep}
{}

\MyMakeRefMacros{\lemref}{\dgLemma}{\dgLemmas}
\MyMakeRefMacros{\Lemref}{\dgCapLemma}{\dgCapLemmas}

\MyMakeDupTheoMacros{corollary}
{\crl}{\ncrl}{\crldup}{\ncrldup}{\dgCorollary}{\crlrep}

\MyMakeRefMacros{\crlref}{\dgCorollary}{\dgCorollaries}
\MyMakeRefMacros{\Crlref}{\dgCapCorollary}{\dgCapCorollaries}

\MyMakeTheoMacros{proposition}{\prp}{\nprp}

{}
\newtheorem*{prp*}{\e{\dgCapProposition}}

{}

\MyMakeRefMacros{\prpref}{\dgProposition}{\dgPropositions}
\MyMakeRefMacros{\Prpref}{\dgCapProposition}{\dgCapPropositions}

\MyMakeDupTheoMacros{my_claim}
{\clm}{\nclm}{\clmdup}{\nclmdup}{\dgClaim}{\clmrep}

\MyMakeRefMacros{\clmref}{\dgClaim}{\dgClaims}
\MyMakeRefMacros{\Clmref}{\dgCapClaim}{\dgCapClaims}

\MyMakeDupTheoMacros{theorem}
{\theo}{\ntheo}{\theodup}{\ntheodup}{\dgTheorem}{\theorep}

\MyMakeRefMacros{\theoref}{\dgTheorem}{\dgTheorems}
\MyMakeRefMacros{\Theoref}{\dgCapTheorem}{\dgCapTheorems}

\MyMakeTheoMacros{postulate}{\postu}{\npostu}

\MyMakeRefMacros{\posturef}{\dgPostulate}{\dgPostulates}
\MyMakeRefMacros{\Posturef}{\dgCapPostulate}{\dgCapPostulates}

\MyMakeTheoMacros{definition}{\defi}{\ndefi}

\MyMakeRefMacros{\defiref}{\dgDefinition}{\dgDefinitions}
\MyMakeRefMacros{\Defiref}{\dgCapDefinition}{\dgCapDefinitions}

\MyMakeTheoMacros{problem}{\prob}{\nprob}

\MyMakeRefMacros{\probref}{\dgProblem}{\dgProblems}
\MyMakeRefMacros{\Probref}{\dgCapProblem}{\dgCapProblems}

\MyMakeTheoMacros{myremark}{\rem}{\nrem}

\MyMakeRefMacros{\remref}{\dgRemark}{\dgRemarks}
\MyMakeRefMacros{\Remref}{\dgCapRemark}{\dgCapRemarks}

\MyMakeTheoMacros{conjecture}{\conj}{\nconj}

\MyMakeRefMacros{\conjref}{\dgConjecture}{\dgConjectures}
\MyMakeRefMacros{\Conjref}{\dgCapConjecture}{\dgCapConjectures}

\renewcommand{\qedsymbol}{$\blacksquare$}

\newcommand{\prfstart}[1][]{\ifthenelse{\equal{}{#1}}%
{\begin{proof}\renewcommand{\qedsymbol}{$\blacksquare$}}%
{\begin{proof}[\dgProofOf\ #1]%
\renewcommand{\qedsymbol}{$\blacksquare_{\mbox{\it{\scriptsize{#1}}}}$}}%
}
\newcommand{\prfend}[1][*]{%
\ifthenelse{\equal{}{#1}}{\renewcommand{\qedsymbol}{$\blacksquare$}}{}%
\ifthenelse{\equal{*}{#1}}{}%
{\renewcommand{\qedsymbol}{$\blacksquare_{\mbox{\it{\scriptsize{#1}}}}$}}%
\end{proof}\renewcommand{\qedsymbol}{$\blacksquare$}%
}

\newcommand{\sect}[2][]{
\ifthenelse{\equal{*}{#2}}
{\section*}
{\ifthenelse{\equal{}{#1}}
{\section{#2}}
{\section{#2}\label{#1}}
}
}

\newcommand{\para}[2][]{\ifthenelse{\equal{}{#1}}
{\paragraph{#2}}
{\paragraph{#2}\label{#1}}}

\MyMakeRefMacros{\chref}{\dgChapter}{\dgChapters}
\MyMakeRefMacros{\Chref}{\dgCapChapter}{\dgCapChapters}

\MyMakeRefMacros{\sref}{\dgSection}{\dgSections}
\MyMakeRefMacros{\Sref}{\dgCapSection}{\dgCapSections}

\MyMakeRefMacros{\ssref}{\dgSubsection}{\dgSubsections}
\MyMakeRefMacros{\Ssref}{\dgCapSubsection}{\dgCapSubsections}

\MyMakeRefMacros{\sssref}{\dgSubsection}{\dgSubsections}
\MyMakeRefMacros{\Sssref}{\dgCapSubsection}{\dgCapSubsections}

\MyMakeRefMacros{\figref}{\dgFigure}{\dgFigures}
\MyMakeRefMacros{\Figref}{\dgCapFigure}{\dgCapFigures}

\newcommand{\IfMathMode}[2]{\ifmmode{#1}\else{#2}\fi}

\newcommand{\Ensuremath}{\ensuremath}

\newcommand{\fbr}[1]{\IfMathMode%
{#1}{$#1$}}                     %

\newcommand{\fnbr}[1]{\mbox{\fbr{#1}}}  %

\newcommand{\fla}[2][*]{\ifthenelse{\equal{}{#1}}{\fbr{#2}}{\fnbr{#2}}}
\newcommand{\f}{\fla}

\newcommand{\malabel}[1]{\addtocounter{equation}{1}\tag{\theequation}\label{#1}}
\newcommand{\mal}[2][]{\MyChangeMathMargins%
\ifthenelse{\equal{}{#1}}%
{\begin{align*} #2 \end{align*}}%
{\ifthenelse{\equal{P}{#1}}%
{\allowdisplaybreaks\begin{align*} #2%
\end{align*}\interdisplaylinepenalty=10000}%
{\begin{align*} \malabel{#1} #2 \end{align*}}%
}%
}

\newcommand{\m}{\mal}

\newcommand{\mac}{\substack}

\MyMakeEqRefMacros{\equref}{\dgEquation}{\dgEquations}
\MyMakeEqRefMacros{\Equref}{\dgCapEquation}{\dgCapEquations}

\MyMakeEqRefMacros{\expref}{\dgExpression}{\dgExpressions}
\MyMakeEqRefMacros{\Expref}{\dgCapExpression}{\dgCapExpressions}

\MyMakeEqRefMacros{\inequref}{\dgInequality}{\dgInequalities}
\MyMakeEqRefMacros{\Inequref}{\dgCapInequality}{\dgCapInequalities}

\newcommand{\bref}[1]{(\ref{#1})}

\newcommand{\lf}{\mathopen{}\mathclose\bgroup\left}
\newcommand{\rt}{\aftergroup\egroup\right}

\providecommand{\middle}{\big}
\newcommand{\md}{\middle}

\newmatop[disc]{\disc}{disc_{#1}}

\makeatletter
\def\moverlay{\mathpalette\mov@rlay}
\def\mov@rlay#1#2{\leavevmode\vtop{%
\baselineskip\z@skip \lineskiplimit-\maxdimen
\ialign{\hfil$\m@th#1##$\hfil\cr#2\crcr}}}
\newcommand{\charfusion}[3][\mathord]{
#1{\ifx#1\mathop\vphantom{#2}\fi
\mathpalette\mov@rlay{#2\cr#3}
}
\ifx#1\mathop\expandafter\displaylimits\fi}
\makeatother

\providecommand{\cupdot}{\charfusion[\mathbin]{\cup}{\cdot}}
\providecommand{\bigcupdot}{\charfusion[\mathop]{\bigcup}{\cdot}}

\newcommand{\h}[2][]{\ifthenelse{\equal{}{#2}}%
{\mathop H_{#1}}%
{\mathop H_{#1}{\l({#2}\r)}}}
\newcommand{\hh}[3][]{\mathop H_{#1}%
{\l({#2}\vphantom{|_1^1}\md|\vphantom{|_1^1}{#3}\r)}}
\newcommand{\hm}[2][]{\ifthenelse{\equal{}{#2}}%
{\mathop {H_{\txt{min}}}_{#1}}%
{\mathop {H_{\txt{min}}}_{#1}{\l({#2}\r)}}}

\newcommand{\KL}[2]{d_{KL}\lf({#1}\md\|\vphantom{|_1^1}{#2}\rt)}

\providecommand{\E}[2][]{\mathop{\pmb{E}}_{#1}\lf[{#2}\rt]}

\newcommand{\PR}[2][]{\mathop{\pmb{Pr}}_{#1}\lf[{#2}\rt]}
\newcommand{\PRr}[3][]{\mathop{\pmb{Pr}}_{#1}\lf[{#2}\vphantom{|_1^1}\md|\vphantom{|_1^1}{#3}\rt]}

\providecommand{\U}{}  %
\renewcommand{\U}[1][]{\ifthenelse{\equal{}{#1}}%
{{\cal U}}%
{{\cal U}_{#1}}}

\newcommand{\GF}[2][]{{\mathcal GF_{#2}^{#1}}}

\providemat{\QQ}{\mathbb{Q}}
\providemat{\NN}{\mathbb{N}}
\providemat{\CC}{\mathbb{C}}
\providemat{\RR}{\mathbb{R}}
\providemat{\ZZ}{\mathbb{Z}}

\newcommand{\ord}[1][]{\nospell{\ifthenelse{\equal{}{#1}}%
{\txt{'th}}%
{\ifthenelse{\equal{1}{#1}}{$1\txt{'st}$}{\ifthenelse{\equal{2}{#1}}{$2\txt{'nd}$}{\ifthenelse{\equal{3}{#1}}{$3\txt{'rd}$}{\fla{#1\txt{'th}}}}}}}}

\newcommand{\fr}[3][*]{%
\ifthenelse{\equal{*}{#1}}%
{\frac{#2}{#3}}{}%
\ifthenelse{\equal{/}{#1}}%
{\dr{#2}{#3}}{}%
\ifthenelse{\equal{}{#1}}%
{\lf.#2\md/#3\rt.}{}%
\ifthenelse{\equal{p_}{#1}}%
{\lf.\lf(#2\rt)\md/#3\rt.}{}%
\ifthenelse{\equal{_p}{#1}}%
{\lf.#2\md/\lf(#3\rt)\rt.}{}%
\ifthenelse{\equal{pp}{#1}}%
{\lf.\lf(#2\rt)\md/\lf(#3\rt)\rt.}{}%
}

\newcommand{\sq}{\sqrt}

\newcommand{\set}[2][]{\ifthenelse{\equal{}{#1}}%
{\Ensuremath{\lf\{#2\rt\}}}%
{\Ensuremath{\lf\{#2\vphantom{|_1^1}\md|\vphantom{|_1^1}#1\rt\}}}}

\newcommand{\sett}[2]{\Ensuremath{\lf\{#1\vphantom{|_1^1}\md|\vphantom{|_1^1}#2\rt\}}}

\newcommand{\Log}[2][]{\ifthenelse{\equal{}{#1}}%
{\log\lf(#2\rt)}%
{\log_{#1}\lf(#2\rt)}%
}

\newcommand{\Min}[2][]{\ifthenelse{\equal{}{#1}}%
{\Ensuremath{\min\lf\{#2\rt\}}}%
{\Ensuremath{\min\lf\{#2\vphantom{|_1^1}\md|\vphantom{|_1^1}#1\rt\}}}}

\newcommand{\Inf}[2][]{\ifthenelse{\equal{}{#1}}%
{\Ensuremath{\inf\lf\{#2\rt\}}}%
{\Ensuremath{\inf\lf\{#2\vphantom{|_1^1}\md|\vphantom{|_1^1}#1\rt\}}}}

\newfunction{\asO}{O}
\newfunction{\asOm}{\Omega}

\newcommand{\sz}[2][]{\ifthenelse{\equal{}{#1}}%
{\Ensuremath{\lf|#2\rt|}}%
{\Ensuremath{\lf|#2\rt|_{#1}}}}
\providecommand{\norm}[2][]{\ifthenelse{\equal{}{#1}}%
{\Ensuremath{\lf\|#2\rt\|}}%
{\Ensuremath{\lf\|#2\rt\|_{#1}}}}

\newcommand{\txt}[1]{\textrm{#1}}  %

\newcommand{\Cl}{\mathcal}  %

\DeclareMathAlphabet{\mathlowcal}{OT1}{pzc}{m}{it}

\newidenT{\Pp}{P}
\newidenT[BPP]{\BPP}{BPP_{#1}}
\newidenT{\ZPP}{ZPP}
\newidenT{\SBP}{SBP}
\newidenT{\coSBP}{coSBP}
\newidenT[RP]{\RP}{RP_{#1}}
\newidenT{\PP}{PP}
\newidenT{\UPP}{UPP}
\newidenT{\coRP}{coRP}
\newidenT{\BQP}{BQP}
\newidenT{\NP}{NP}
\newidenT{\coNP}{coNP}
\newidenT{\AM}{AM}
\newidenT[MA]{\MA}{MA_{#1}}
\newidenT[QMA]{\QMA}{QMA_{#1}}
\newidenT[ZAM]{\ZAM}{ZAM_{#1}}
\newidenT[UAM]{\UAM}{UAM_{#1}}
\newidenT{\PH}{PH}
\newidenT{\PSPACE}{PSPACE}
\newidenT{\EXP}{EXP}
\newidenT{\NEXP}{NEXP}

\newidenT{\DNF}{DNF}
\newidenT{\Eq}{Eq}
\newidenT{\Disj}{Disj}
\newidenT{\IP}{IP}

\newmat{\mset}{\smin\set}

\newcommand{\nin}{\not\in}  %

\newcommand{\Then}{\Longrightarrow}

\newcommand{\tm}{\cdot}
\newcommand{\smin}{\setminus}
\newcommand{\eps}{\varepsilon}
\newcommand{\deq}{\stackrel{\textrm{def}}{=}}

\newcommand{\ds}[1][]
{\ifthenelse{\equal{}{#1}}{\allowbreak\dots}{#1\allowbreak\dots#1}}
\newmat{\dc}{\ds[,]}
\newmat{\dcirc}{\ds[\circ]}

\mydef{\01}{\set{0,1}}

\mydef{\l(}{\lf(}
\mydef{\r)}{\rt)}

\mydef{\+}{\oplus}

\mathchardef\myhyphen="2D
\mydef{\-}{\myhyphen}

\newcommand{\abstart}{\begin{abstract}}
\newcommand{\abend}{\end{abstract}}

{} %

{}  %

\protected \def \dg #1{%
\textcolor{Red}
{
{\normalmarginpar\marginnote{\bl{DG's comment}}}
{\reversemarginpar\marginnote{\bl{DG's comment}}\\}
\IfMathMode{
~~~\txt{#1}~
}{
~\\~~~#1~\\
{\normalmarginpar\marginnote{\bl{\ul{------}}}}
{\reversemarginpar\marginnote{\bl{\ul{------}}}\\}
}
}
\ClassWarning{My Macros}{#1}
}

\newcommand{\fn}[2][]{%
\IfMathMode{}{}%
\ifthenelse{\equal{}{#1}}%
{\footnote{#2}}%
{\footnote{\label{#1}#2}}%
}



\DeclareTextFontCommand{\bemph}{\bfseries}
\DeclareTextFontCommand{\ibemph}{\bfseries\em}
{} %
\newcommand{\e}{\emph}

{}  %

\newcommand{\bl}[1]{{\bf #1}} %

\newcommand{\il}[1]{{\it #1}} %

\providecommand{\ul}[1]{\underline{#1}} %

\newcommand{\tbb}{\qquad}
\newcommand{\tbbb}{\qquad\qquad}

\makeatother%

{}

\newcommand{\MyChangeMathMargins}{%
}

{}

\makeatother%

\newidenT{\SV}{SV}
\newident{\SVd}{SV_\delta}

\newident{\Hamd}{Ham_d}
\newident{\Hamdi}{Ham_d^i}

\title{
Santha-Vazirani sources, deterministic condensers\\
and very strong extractors}

\newcommand{\instDGPP}{Institute of Mathematics of the Czech Academy of Sciences, \v Zitna 25, Praha 1, Czech Republic.}

\newcommand{\thanksDGPP}{Partially funded by the grant 19-27871X of GA \v CR.}

\newcommand{\thanksDGonly}{Part of this work was done while visiting the Centre for Quantum Technologies at the National University of Singapore, and was partially supported by the Singapore National Research Foundation, the Prime Minister's Office and the Ministry of Education under the Research Centres of Excellence programme under grant R 710-000-012-135.}

\author{Dmitry Gavinsky\thanks{\instDGPP\newline\thanksDGPP} \thanks{\thanksDGonly}
\and Pavel Pudl\'ak\protect\footnotemark[1]
}

\begin{document}

\maketitle

\thispagestyle{empty}

\abstart
The notion of \e{semi-random sources}, also known as \e{Santha-Vazirani (\SV)} sources, stands for a sequence of $n$ bits, where the dependence of the \ord[i] bit on the \e{previous $i-1$ bits} is limited for every $i\in[n]$.
If the dependence of the \ord[i] bit on the \e{remaining $n-1$ bits} is limited, then this is a \e{strong \SV-source}.
Even the strong \SV-sources are known not to admit (universal) \e{deterministic extractors}, but they have \e{seeded extractors}, as their min-entropy is \asOm n.

It is intuitively obvious that strong \SV-sources are \e{more than just high-min-entropy sources}, and this work explores the intuition.
\e{Deterministic condensers} are known not to exist for general high-min-entropy sources, and we construct for any constants $\eps,\delta\in(0,1)$ a deterministic condenser that maps $n$ bits coming from a strong \SV-source with \e{bias} at most $\delta$ to \asOm n\ bits of \e{min-entropy rate} at least $1-\eps$.

In conclusion we observe that deterministic condensers are closely related to \e{very strong extractors} -- a proposed strengthening of the notion of \e{strong (seeded) extractors}:\ in particular, our constructions can be viewed as \e{very strong extractors for the family of strong Santha-Vazirani distributions}.
The notion of very strong extractors requires that the output remains unpredictable even to someone who knows not only the seed value (as in the case of strong extractors), but also the extractor's outputs corresponding to the same input value with each of the preceding seed values (say, under the lexicographic ordering).
Very strong extractors closely resemble the original notion of \SV-sources, except that the bits must satisfy the unpredictability requirement only on average.
\abend

\setcounter{page}{0}
\newpage

\sect[s_intro]{Introduction}

According to the principles of quantum mechanics, perfectly unbiased and independent random bits can be generated in a physical experiment; however, the imperfectness of practical implementations makes it impossible to deduce from the postulates of quantum mechanics perfect independence of the generated bit sequences.
Another possibility is using various chaotic system as a source of of random bits, but in this case it also remains unclear whether perfect (or arbitrarily close to such) independence can be claimed.

The natural question then is: can we reduce the bias and the dependence to a negligible minimum by post-processing the generated bits, coming from a non-perfect source?
The post-processing must, of course, be done by a deterministic algorithm.
The answer, surprisingly, is that this is not possible (at least in some models of the real situation).

In 1986 Santha and Vazirani~\cite{SV86_Gen} have defined and studied the notion of \e{semi-random sources}, now known as \SV-sources.
These are sequences of $n$ bits $X_1\dc X_n$, whose values cannot be accurately predicted in the following sense:\ for some $\delta\in[0,1)$ and any $z\in\01^{i-1}$ for $i\in[n]$, it holds that
\m{
\fr{1-\delta}2 \le \PRr{X_i=1}{X_1=z_1\dc X_{i-1}=z_{i-1}}
\le \fr{1+\delta}2
.}
If $\delta\in(0,1)$ is fixed, we will denote the class of such sources by $SV_\delta$.

In 2004 Reingold, Vadhan and Wigderson~\cite{RVW04_A_No} defined an even stronger class of entropic bits, which they called \e{strong Santha-Vazirani sources}, where for some $\delta\in[0,1)$ and any $z\in\01^n$:
\m{
\fr{1-\delta}2
&\le \PRr{X_i=1}{X_1=z_1\dc X_{i-1}=z_{i-1}, X_{i+1}=z_{i+1}, X_n=z_n}\\
&\le \fr{1+\delta}2
}
for any $i\in[n]$. Similarly, for a fixed $\delta\in(0,1)$, we will call such sources strong $SV_\delta$.

A \e{deterministic extractor} is a function that maps $n$ bits to (at least) $1$ bit that is nearly-unbiased, as long as the input bits are coming from the corresponding type of entropy source.
Santha and Vazirani demonstrated~\cite{SV86_Gen} that \SV-sources did not admit a (universal) deterministic extractor.
Later Reingold, Vadhan and Wigderson~\cite{RVW04_A_No} generalized this impossibility result to the case of strong \SV-sources.

It is well known, on the other hand, that \e{seeded extractors} (see Sect.~\ref{s_prelim}) exist for the class of bit sources whose min-entropy is \asOm n; since the \SV-sources, obviously, belong to that class,\fn
{
Throughout the work we will assume, unless stated otherwise, $\delta\in\asOm1$ in the context of \SV-sources.
}
seeded extractors exist, in particular, for them.

The definition of \SV-sources is very natural,
so it is both important and interesting to investigate this type of randomness.
Intuitively, it is obvious that \SV-sources -- especially the strong form -- %
are more structured than general sources with the same min-entropy and one should be able to use this fact. 
This work will explore this intuition.

\e{Deterministic condensers} are functions that map $n$ bits to $m~(<n)$ bits, such that  the \e{min-entropy rate} of the output is higher than the min-entropy rate of the input. (Min-entropy rate is the ratio of the min-entropy to the number of bits.) This is certainly not always possible, for instance, if the input has the maximum min-entropy. Typically we have a class of sources and a lower bound on their min-entropy rate and we want to achieve higher lower bound on the min-entropy of the output distributions.
Similarly to deterministic extractors, deterministic condensers do not exist for the class of \emph{all sources} whose min-entropy rate is at least a given number. 

In this paper we prove two results about SV and strong SV sources. First we show that for $SD_\delta$ sources, there are no non-trivial condensers, which means that in general we cannot improve the min-entropy rate using condensers. On the other hand, we construct for any constants $\eps,\delta\in(0,1)$, a deterministic condenser that maps $n$ bits coming from a \emph{strong} \SV$_\delta$ source to \asOm n\ bits of min-entropy rate at least $1-\eps$.

As deterministic condensers are somewhat exotic objects (primarily due to their non-existence for the general class of high-min-entropy sources), this work continues by investigating that notion.

The familiar notion of \e{strong (seeded) extractors} can be strengthened further -- we call the new type of distribution-transforming objects \e{very strong extractors} -- here the output must remain unpredictable even to someone who knows not only the seed value (as in the case of strong extractors), but also the extractor's outputs corresponding to the same input value with each of the preceding seed values (see Sect.~\ref{s_prelim}).
We show that deterministic condensers can be easily transformed into very strong extractors;\fn
{
The reverse transformation is almost possible:\ the resulting mapping can only guarantee high \e{entropy rate} of the output (see Sect.~\ref{s_conc}).
}
therefore, deterministic condensers are ``stronger'' objects than strong extractors, but ``weaker'' than deterministic extractors.
Via the same transformation, the main construction of this work gives a very strong extractor for the class of strong \SV-sources.

\sect[s_prelim]{Preliminaries}

For an excellent survey of error correcting codes, see~\cite{MS77_The_The}.

Let $[n]=\set{1\dc n}$ and $\ZZ_n=\ZZ/n\ZZ\simeq\set{0\dc n-1}$.
Let $\log$ be base-$2$ by default.

For $x\in\01^n$ and $i\in [n]$, we will use both $x_i$ and $x(i)$ to address the \ord[i] bit of $x$.
Let $\sz x$ denote the Hamming weight of $x$.
For $y\in\01^n$, let $x\+ y$ denote the bit-wise XOR of the two vectors.
For $y\in\01^m$, let $x\circ y\in\01^{n+m}$ denote the corresponding concatenation.

We will often implicitly assume the arithmetic of $\GF2$ and its generalization to $\GF[n]2$ (e.g., $x\+ y$ is the two vectors' sum in the \f n-dimensional linear space).
In the context of $\GF[n]2$ for $i\in[n]$ we will write $e_i$ to denote the \ord[i] unit vector in $\GF[n]2$ and let $e_0\deq\bar0\in\GF[n]2$.

For sets $A$ and $B$ we will write $A\cupdot B$ to denote the union while implying the sets' disjointness (the notation, especially the indexed version ``$\bigcupdot_{i=1}^n\ds$'', is a convenient way of addressing partitions).

For a non-empty finite set $A$ we will denote by $\U[A]$ the uniform distribution on $A$.
Let $\mu$ and $\nu$ be distributions on $A$, we will say that they are \e{\f\eps-close} if the $l_1$-distance between them is at most $2\eps$.

Let $X$ be a random variable, then
\m{
\hm\mu = \hm[X\sim\mu]X \deq \Min[a\in A]{\Log{\fr1{\mu(a)}}}
}
is the \e{min-entropy} of $\mu$ and
\m{
\h\mu = \h[X\sim\mu]X \deq \sum_{a\in A}\mu(a)\tm\Log{\fr1{\mu(a)}}
}
is the \e{entropy} of $\mu$.
If $A=\01^n$, then $\fr{\hm\mu}n$ is the \e{min-entropy rate} and $\fr{\h\mu}n$ is the \e{entropy rate} of $\mu$.

It is intuitively obvious that a distribution is close to the uniform on its support if and only if the entropy of the distribution is close to the maximum (the logarithm of its support size).
The following statement formalizes this intuition.

\fct[c_mu_U]{Let $\mu$ be a distribution supported on $A$.
Then
\m{
\fr{\log e}2\tm\norm[1]{\mu-\U[A]}^2
&~\le~ \Log{|A|}-\h\mu\\
&\hspace{-48pt}\le~ \fr12\tm\norm[1]{\mu-\U[A]}\tm\Log{|A|}
+ \sq{2\log e\tm\norm[1]{\mu-\U[A]}\tm\Log{|A|}}
.}
}

\prfstart
Let $d\deq\norm[1]{\mu-\U[A]}$.
By Pinsker's inequality, see~\cite{P60},
\m{
\fr{\log e}2\tm d^2
& \le \KL\mu{\U[A]}
= \sum_{a\in A}\mu(a)\tm\Log{\fr{\mu(a)}{\U[A](a)}}\\
& = \sum_{a\in A}\mu(a)\tm \l(\Log{|A|}-\Log{\fr1{\mu(a)}}\r)
= \Log{|A|}-\h\mu
,}
which establishes the desired lower bound on $\Log{|A|}-\h\mu$.

For every $t\ge0$ let
\m[m_B_t]{
B_t \deq \sett{x\in A}{\mu(x) > \fr{1+t}{|A|}}
,}
then
\m[m_U_B_t]{
\U[A](B_t) \le \fr{\mu(B_t)}{1+t}
&\Then~
\fr d2 \ge \mu(B_t)-\U[A](B_t) \ge \fr t{1+t}\tm\mu(B_t)\\
&\Then~
\mu(B_t) \le \fr d2+\fr d{2t}
.}
So,
\m{
\h\mu
&\ge \sum_{x\nin B_t}\mu(x)\tm\Log{\fr1{\mu(x)}}
\ge \mu(A\smin B_t)\tm\Log{\fr{|A|}{1+t}}\\
&\ge \Log{|A|}\tm\l(1-\fr d2-\fr d{2t}\r) - \Log{1+t}\\
&\ge \Log{|A|} - \fr{d\tm\Log{|A|}}2 - \fr{d\tm\Log{|A|}}{2t} - t\tm\log e
,}
where the second inequality is~\bref{m_B_t}, the third is~\bref{m_U_B_t}, and the last one holds, as $\Log{1+t}\le t\tm\log e$.
Choosing $t=\sq{\fr{d\tm\Log{|A|}}{2\log e}}$ establishes the desired upper bound on $\Log{|A|}-\h\mu$.
\prfend

The following are several families of \e{discrete distributions} that we will be interested in.

\ndefi[def_SV]{Santha-Vazirani sources}{Let $\delta\in[0,1)$ and $X=X_1\dc X_n$ be a random variable distributed over $\01^n$ according to a distribution $\mu$.
If
\m{
\forall i\in[n],\,z\in\01^{i-1}:\:
\PRr[\mu]{X_i=1}{\bigwedge_{j=1}^{i-1}X_j=z_j}\in\lf[\fr{1-\delta}2,\fr{1+\delta}2\rt]
,}
then we call $X$ a \e{Santha-Vazirani source with bias $\delta$ (\SVd)} and $\mu$ a \e{Santha-Vazirani distribution with bias $\delta$}.
If
\m{
\forall i\in[n],\,z\in\01^n:\:
\PRr[\mu]{X_i=1}{\bigwedge_{j\in[n]\mset i}X_j=z_j}\in\lf[\fr{1-\delta}2,\fr{1+\delta}2\rt]
,}
then we call $X$ a \e{strong Santha-Vazirani source with bias $\delta$} and $\mu$ a \e{strong Santha-Vazirani distribution with bias $\delta$}.
}

The following are several types of \e{distribution transformations} that we will be interested in.

\ndefi[def_cond]{Deterministic condensers}{Let $\Cl F$ be a family of distributions over $\01^n$.
A function $h:\01^n\to\01^m$ is a \e{deterministic $k$-condenser for $\Cl F$} if
\m{
\hm[X\sim\mu]{h(X)}\ge k
}
for every $\mu\in\Cl F$.

The \e{min-entropy rate} of the condenser is $\dr km$.
The condenser is \e{non-trivial} if
\m{
\fr km > \Inf[\mu\in\Cl F]{\fr{\hm{\mu}}n}
.}
}

In the concluding~\sref{s_conc} we will consider the \e{entropy rate} of a condenser, defined as
\m[m_en_ra]{
\fr{\Inf[\mu\in\Cl F]{\h[X\sim\mu]{h(X)}}}m
,}
which is, obviously, at least as high as the min-entropy rate of the same condenser.
An object, analogous to a deterministic condenser, but only guaranteeing certain entropy rate (as opposed to min-entropy) will be called \e{entropy condenser}.

The following statement must be folklore.

\fct[f_Cond_h_min]{Let $n\ge m\in\NN$ and $l\in[0,n]$.
No non-trivial deterministic condenser from $\01^n$ to $\01^m$ exists for the family of all distributions whose min-entropy is at least $l$.
}
In Theorem~\ref{no-sv-condenser} below we will prove the non-existence of non-trivial condenser even for a restricted class of distributions whose min-entropy rate is at least some bound $l$, namely for the distributions of $SV_\delta$ sources (where $l=\frac 2{1+\delta}$).

\ndefi[def_det_extr]{Deterministic extractors}{Let $\eps\in[0,1]$ and $\Cl F$ be a family of distributions over $\01^n$.
A function $h:\01^n\to\01^m$ is a \e{deterministic \f\eps-extractor for $\Cl F$} if the distribution of $h(X)$ is \f\eps-close to $\U[\01^m]$ when $X\sim\mu$ for any $\mu\in\Cl F$.
}

It is well-known (and follows from \fctref{f_Cond_h_min}) that no non-trivial deterministic extractor exists for the family of distributions whose min-entropy is at least $l$.
It was shown in~\cite{RVW04_A_No} that no non-trivial deterministic extractor exists for the family of strong \SV-sources even for $m=1$, which implies that no such extractors exist for any $m$. 
On the other hand, there are known constructions of \e{seeded extractors} for the family of high-min-entropy distributions.

\ndefi[def_seed_extr]{Seeded extractors}{Let $\eps\in[0,1]$ and $\Cl F$ be a family of distributions over $\01^n$.

A function $h:\01^n\times[D]\to\01^m$ is a \e{seeded \f\eps-extractor for $\Cl F$} if the distribution of $h(X,S)$ is \f\eps-close to $\U[\01^m]$ when $S\sim\U[{[D]}]$ and $X\sim\mu$ for any $\mu\in\Cl F$.
}

Of the following two versions that strengthen the notion of seeded extractors, the first is standard and the second is, to the best of our knowledge, new.

\ndefi[def_extr_strong]{Strong and very strong seeded extractors}{Let $\eps\in[0,1]$, $\Cl F$ be a family of distributions over $\01^n$ and $h:\01^n\times[D]\to\01^m$.

For $s\in[D]$ let $\nu_s^\mu$ be the distribution of $h(X,s)$ when $X\sim\mu$ -- that is, the distribution of $h(X,S)$, conditioned on $[S=s]$.
For $y_1\dc y_{s-1}\in\01^m$ let $\nu_{s,y_1\dc y_{s-1}}^\mu$ be the distribution of $h(X,S)$ when $X\sim\mu$, conditioned on $[S=s, h(X,1)=y_1\dc h(X,s-1)=y_{s-1} ]$ (let it be undefined if the conditioning is inconsistent).

We call $h$ a \e{strong \f\eps-extractor for $\Cl F$} if
\m{
\E[{T\sim\U[{[D]}]}]
{\norm[1]{\nu_T^\mu-\U[\01^m]}} \le 2\eps
.}

We call $h$ a \e{very strong \f\eps-extractor for $\Cl F$} if
\m{
\E[{T\sim\U[{[D]}],\, Z\sim\mu}]
{\norm[1]{\nu_{T,h(Z,1)\dc h(Z,T-1)}^\mu-\U[\01^m]}} \le 2\eps
.}
}

In other words, a seeded extractor is \e{strong} if its output $h(X,S)$ remains unpredictable even when the seed value $S$ is known, and it is \e{very strong} if the output remains unpredictable even when the seed value $S$, as well as the outputs corresponding to the preceding seed values ($h(X,1)\dc h(X,S-1)$) are known.

The concept of very strong extractors somewhat resembles the original notion of \SV-sources:\ if we write down the sequence $(h(X,1)\dc h(X,D))$, then \e{most} (not necessarily all, as would be the case for \SV-sources) of the blocks will be sufficiently unpredictable, even conditioned on the previous blocks.
Very strong extractors will be compared to deterministic condensers in the concluding part of this work (Sect.~\ref{s_conc}).

\section{No condensers for SV sources}

It \cite{SV86_Gen} Santha and Vazirani proved that in general it is not possible to extract a bit from $SV_\delta$ sources that is biased less than $\delta$.
Later it was shown~\cite{RVW04_A_No} that this is also true for strong \SVd-sources. This implies that there is no extractor that would produce a distribution that is guaranteed to have the statistical distance from uniform less than $\delta$. In other words, there are no $\epsilon$-extractors for strong $SV_\delta$ sources with $\epsilon<\delta$.

In this section we prove another impossibility result for \SV-sources (we will see in \sref{s_str_con} that the same is not true for strong \SVd-sources).

\begin{theorem}\label{no-sv-condenser}
There is no non-trivial min-entropy condenser for \SVd-sources for any $0<\delta<1$.
Namely,
\begin{enumerate}
\item the min-entropy-rate of any \SVd-source is at least $\log\frac 2{1+\delta}$,
and
\item for every $F:\{0,1\}^n\to\{0,1\}^m$, there exists an \SVd-source $X$, such that the min-entropy-rate of $F(X)$ is at most $\log\frac 2{1+\delta}$.
\end{enumerate}
\end{theorem}

\prfstart[\theoref{no-sv-condenser}]
1. Let $X$ be an \SVd-source. Then by definition,
\m{
\PRr{X_i=a_i}{X_1=a_1,\dots,X_{i-1}=a_{i-1}} \leq \fr{1+\delta}2
}
for every $i$ and $a_1\dc a_{i-1},a_i\in\01$.
Hence
\m{
\PR{X=a} = \prod_{i=1}^n\PRr{X_i=a_i}{X_1=a_1,\dots,X_{i-1}=a_{i-1}}
\le \l(\fr{1+\delta}2\r)^n
}
for every $a\in\01^n$, which proves that the min-entropy rate of $X$ is $\geq\frac 2{1+\delta}$.

2. To prove the second claim, we need the following lemma.

\begin{lemma}\label{l_A}
For every $\delta\in(0,1)$ and any non-empty $A\subseteq\{0,1\}^n$, there exists an \SVd-distribution $\mu$, such that 
\m[e1]{
\mu(A) \ge \l(\fr{1+\delta}2\r)^{n-\log|A|}
.}
\end{lemma}

First we prove 2., assuming that the lemma is valid.

Let $s$ be such that $|F^{-1}(s)|\geq 2^{n-m}$.
Then \lemref{l_A} implies that for some \SVd-source $X$ it holds that
\m{
\PR{F(X)=s} \ge \l(\frac{1+\delta}2\r)^{n-(n-m)}
= \l(\frac{1+\delta}2\r)^m
.}
Hence
\m{
H_\infty[F(X)]\leq m\cdot\log\frac 2{1+\delta}
,}
as required.
\prfend

Now let us prove the lemma.

\prfstart[\lemref{l_A}]
Denote by $p=\frac{1-\delta}2$, $q=\frac{1+\delta}2$. 
The idea is simple: put as much weight as possible to $A$. So we define an $SD_\delta$ source $X$ by

\smallskip
\[\begin{array}{lll}
\PRr{X_i=0}{(X_1,\dots, X_{i-1})=u}&= q&\mbox{ if }
|\{v\ |\ u0v\in A\}|\geq |\{v\ |\ u1v\in A\}|\\
\\
&=p&\mbox{otherwise.}
\end{array}\]
We will prove (\ref{e1}) by induction on $n$. The base case is clear.

If all $a\in A$ start with 0 or all start with 1, the induction step is trivial. So suppose that this is not the case. Let $A_0$, respectively $A_1$, be those that start with~0, respectively with~1. Assume without loss of generality that $|A_0|\geq|A_1|$. We define $X$ so that $\PR{X_1=0}=q$ and for the conditional probabilities we will use the two sources that maximize the probabilities that $0v\in A_0$ and $1v\in A_1$. If we denote by $a=|A_0|$ and $b=|A_1|$, then for the induction step, it suffices to prove the following inequality
\begin{equation}\label{e2}
q\cdot q^{n-\log a}+p\cdot q^{n-\log b}\geq q^{n+1-\log(a+b)}.
\end{equation}
This is equivalent to
\begin{equation}\label{e3}
q\cdot a^{-\log q}+ p\cdot b^{-\log q}\geq q\cdot(a+b)^{-\log q}.
\end{equation}

Let $a$ be fixed and consider the function 
\[
f(x):= q\cdot a^{-\log q}+ p\cdot x^{-\log q}-q\cdot(a+x)^{-\log q}
\]
in the domain $0\leq x\leq a$. We need to prove that $f(x)$ is non-negative in this domain. To this end, it suffices to prove:
\begin{enumerate}
\item $f(0)=0$, which is immediate,
\item $f(a)=(q+p)a^{-\log q}-q(2a)^{-\log q}=a^{-\log q}-q2^{-\log q}a^{-\log q}=0$,
\item $f'(0)>0$ or $f'(a)<0$,
\item $f'(x)$ has a unique root.
\end{enumerate}

Concerning 3., both are true, but it suffices to prove one of these inequalities. We will check the second one.
\[\begin{array}{rl}
f'(a)=&p(-\log q)a^{-\log q-1}-q(-\log q)(2a)^{-\log q-1}\\ \\
=&(-\log q)a^{-\log q-1}(p-q\cdot2^{-\log q-1})\\ \\
=&(-\log q)a^{-\log q-1}(p-\frac 12)<0,
\end{array}\]
because $-\log q>0$, as well as $a^{-\log q-1}>0$, and $p<\frac 12$.

Concerning 4., 
\[\begin{array}{ll}
f'(x)=0&\Leftrightarrow\\ \\
p(-\log q)x^{-\log q-1}=q(-\log q)(a+x)^{-\log q-1}&\Leftrightarrow\\ \\
p^{\frac 1{-\log q-1}}x=q^{\frac 1{-\log q-1}}(a+x)&\Leftrightarrow\\ \\
(p^{\frac 1{-\log q-1}}-q^{\frac 1{-\log q-1}})x=q^{\frac 1{-\log q-1}}a.
\end{array}\]
Since $p\neq q$, the coefficient at $x$ is non-zero and consequently the equation has a unique solution. This finishes the proof of the inequality~(\ref{e2}) which was needed for the induction step.
\prfend

\sect[s_str_con]{A condenser for strong SV sources}

In this section we will show that in contrast to the standard SV sources, non-trivial deterministic condensers do exist for \emph{strong} SV sources.

\subsection{The construction}\label{s_constr}

In this part we construct a family of functions, which will be shown to act as deterministic condensers for strong Santha-Vazirani sources in the next part.

The following definition is, essentially, due to~\cite{H50_Err}.

\ndefi[def_Ham]{Hamming code}{Let $d\in\NN$ and $M_d\in\01^{d\times(2^d-1)}$ be the matrix whose columns are the numbers $1\dc2^d-1$ in their binary representation.
Then the set
\m{
\Hamd \deq \sett{x\in\01^{2^d-1}}{M_d\tm x=\bar0}
}
is the \e{Hamming code} of length $2^d-1$.
}

The Hamming codes are known to have minimum distance $3$, which is easy to see:
\m[m_Hdist_3]{
x_1\neq x_2 \in \Hamd
~\Then~
M_d \tm (x_1\+ x_2) = \bar0
~\Then~
\sz{x_1\+ x_2} \ge 3
,}
as the columns of $M_d$ are linearly independent in $\GF[d]2$.

For $i\in\ZZ_{2^d}$, let
\m[m_Hamdi]{
\Hamdi \deq \sett{x\+e_i}{x\in\Hamd}
.}
From \bref{m_Hdist_3} it follows that these $2^d$ sets are pairwise disjoint.
As $\sz{\Hamd}=2^{2^d-1-d}=2^{2^d-1}/2^d$, 
\m[m_part]{
\01^{2^d-1} = \bigcupdot_{i=0}^{2^d-1} \Hamdi
.}
This is the well-known fact that Hamming codes are \emph{perfect}.
Let
\m[m_g_d]{
g_d:\: \01^{2^d-1}\to\01^d
}
point to the equivalence class of its argument:\ namely, $g_d(x)$ is the \e{binary representation} of such $i$ that $x\in\Hamdi$; from~\bref{m_part} it follows that $g_d$ is well-defined.

Let $n=k\tm(2^d-1)$ be a multiple of $2^d-1$, define $f_d:\01^n\to\01^{\fr d{2^d-1}\tm n}$ as follows:
\m[m_f_d]{
f_d(y_1\dc y_k) \deq g_d(y_1)\dcirc g_d(y_k)
,}
where every $y_i$ is a block of length $2^d-1$.

\subsection{Analysis}\label{s_analy}

Following \cite{RVW04_A_No}, we will call a distribution $\nu$ \emph{$\delta$-imbalanced} if for every $x,y$, 
\[
\PR{X=x} \le \frac{1+\delta}{1-\delta}\tm\PR{Y=y}.
\]

\begin{lemma}
If $X$ is a strong $SV_\delta$ source, then the distribution of $g_d(X)$ is $\delta$ imbalanced.
\end{lemma}
\begin{proof}
As $X\sim\mu$ is a strong \SVd-source (see Def.~\ref{def_SV}), for all $x\in\01^{2^d-1}$ and $i\in[2^d-1]$,
\smallskip
\begin{equation}\label{m_hyp_edge}
\fr{\mu(x\+e_i)}{\mu(x)} \le \fr{1+\delta}{1-\delta}.
\end{equation}

Since the Hamming code is a perfect code of distance 3, for every $i\neq j$, there exists $k$ such that $e_i+e_j+e_k\in\Hamd$. Hence
\begin{equation}\label{m_eio}
\Hamd^i=\Hamd^j+e_k.
\end{equation}
Now we can write for $u_1\in\Hamd^i,u_2\in\Hamd^j,$ 
\m[m_colours]{
\mu(g_d^{-1}(u_1))
& = \sum_{x\in\Hamd^i}\mu(x)
\leq \sum_{x\in\Hamd^i}\fr{1+\delta}{1-\delta} \tm\mu(x \+ e_k)\\
& \le \fr{1+\delta}{1-\delta} \tm \sum_{x\in\Hamd^j}\mu(x)
= \fr{1+\delta}{1-\delta} \tm \mu(g_d^{-1}(u_2))
,}
where the inequality is~\bref{m_hyp_edge}. Thus $g_d(X)$ is $\delta$-imbalanced. 
\end{proof}

\begin{lemma}
Let $\nu$ be $\delta$-imbalanced distribution on $\{0,1\}^d$. Then 
\begin{equation}\label{el3}
\hm{\nu} \ge d - \Log{\fr{1+\delta}{1-\delta}}.
\end{equation}
\end{lemma}
\begin{proof}
As $\sum_{u\in\01^d} \nu(u) = 1$, 
there exists some $u_0\in\01^d$, such that $\nu(u_0) \le 2^{-d}$.
Since  $\nu$ is $\delta$-imbalanced, for all $u\in\01^d$,
\[
\nu(u) \le \fr{1+\delta}{1-\delta} \tm \nu(u_0)
\le \fr{1+\delta}{1-\delta} \tm 2^{-d},
\]
which proves (\ref{el3}).
\end{proof}

From the two, lemmas we get
\m[m_hm_gd]{
\hm[X\sim\mu]{g_d(X)} \ge d - \Log{\fr{1+\delta}{1-\delta}}
.}

Now let $n=k\tm(2^d-1)$ for $k\in\NN$ and $Y=(Y_1\dc Y_k)\in\01^n$ be sampled according to a strong \SVd-distribution $\nu$ (every $Y_i$ consists of $2^d-1$ bits).
We want to analyze the resulting distribution of $f_d(Y)\in\01^{k\tm d}$.

Fix any $i\in[k]$ and $y_1\dc y_i\in\01^{2^d-1}$.
By the definition of (strong) \SV-distributions (Def.~\ref{def_SV}), the distribution of $Y_i$ remains strong \SVd\ when conditioned upon $[Y_1=y_1\dc Y_{i-1}=y_{i-1}]$; accordingly, from~\bref{m_hm_gd} it follows that
\m{
\PRr[Y\sim\nu]{Y_i=y_i}{Y_1=y_1\dc Y_{i-1}=y_{i-1}}
\le \fr{1+\delta}{1-\delta} \tm 2^{-d}
.}
Via the trivial induction this implies
\m[m_hm_fd]{
\hm[Y\sim\nu]{f_d(Y)} \ge kd - k\tm\Log{\fr{1+\delta}{1-\delta}}
.}

Therefore, $f_d$ is a deterministic condenser for the family of strong \SVd-distributions, whose min-entropy rate is lower-bounded by
\m{
1 - \fr1d \tm \Log{\fr{1+\delta}{1-\delta}}
.}
Since the min-entropy of a strong \SVd-distribution over $\01^n$ can be as low as $n\tm\Log{\fr2{1+\delta}}$ (as witnessed by the mutually independent distribution of $n$ bits, each taking value ``$1$'' with probability $\fr{1+\delta}2$), the min-entropy rate of a strong \SVd-distribution over $\01^n$ can be as low as $\Log{\fr2{1+\delta}}$ and therefore, $f_d$ is a \e{non-trivial deterministic condenser for strong \SVd-distributions} as long as $\delta\in(0,1)$.

We have established the following (via an explicit construction).

\theo[th_det_cond]{Let $d\in\NN$ and $n$ be a multiple of $2^d-1$.
A deterministic condenser for strong \SV-distributions exists that maps $n$ bits to $\fr{n\tm d}{2^d-1}$ bits and when the input distribution is a strong \SVd\ for $\delta\in[0,1)$ (whose min-entropy rate can be as low as $\Log{\fr2{1+\delta}}$), the generated min-entropy rate is at least
\m{
1 - \fr1d \tm \Log{\fr{1+\delta}{1-\delta}}
.}

In particular, for any constants $\eps,\delta\in(0,1)$, a deterministic condenser for strong \SVd-distributions exists that maps $n$ bits to \asOm{n} bits of min-entropy rate at least $1-\eps$.
}

\sect[s_conc]{Deterministic condensers vs.\ very strong extractors}

As \e{deterministic condensers} are known not to exist for the family of high-min-entropy distributions (\fctref{f_Cond_h_min}), they are not considered in the literature very often.
We conclude this work by discussing these elegant and natural objects.

\para{Very strong extractors from deterministic condensers.}
Recall the notion of \e{entropy rate} (as opposed to \e{min}-entropy rate) of a condenser, given by~\bref{m_en_ra}.

\clm[cl_cond_extr]{
Let $\eps\in[0,1]$, $m<n$, $D$ be such that $D|m$ and $\fr{\ln 2}2\eps\le\fr Dm$. Let $\Cl F$ be a family of distributions over $\01^n$ and let $h:\01^n\to\01^m$ be a deterministic entropy condenser for $\Cl F$  of entropy rate at least $1-\eps$. Then $g:\01^n\times[D]\to\01^{\dr mD}$, defined as
\medskip
\[
g(x,s) = h(x)\upharpoonright_{\fr{(s-1)m}D+1,\dots,\fr{sm}D},
\]
i.e., $h(x)$ restricted to bits $\fr{(s-1)m}D+1,\dots,\fr{sm}D$, is a very strong $\delta$-extractor for $\Cl F$, where
$\delta=\sq{\fr{\ln 2}2\eps\tm\fr mD}$.

In particular, for $D=m$ this gives $g:\01^n\times[m]\to\01$, which is a very strong $\sq{\fr{\ln 2}2\eps}$-extractor for $\Cl F$.
}

As the entropy rate is always at least as high as the min-entropy rate, it follows from the above statement that the construction of \sref{s_constr}, as summarized by \theoref{th_det_cond}, also gives \e{very strong extractors for the family of strong Santha-Vazirani distributions}, namely:

\crl{
For any constants $\eps,\delta\in(0,1)$, a very strong single-bit $\eps$-extractor of seed length $\log n-\asO1$ exists for strong \SVd-distributions.}

\prfstart[\clmref{cl_cond_extr}]
Let $\mu\in\Cl F$.
Similarly to \defiref{def_extr_strong}, for every $s\in[D]$ and $y_1\dc y_{s-1}\in\01^{\dr mD}$ let $\nu_{s,y_1\dc y_{s-1}}^\mu$ be the distribution of $g(X,S)$ when $X\sim\mu$, conditioned on $[S=s, g(X,1)=y_1\dc g(X,s-1)=y_{s-1} ]$ (let it be undefined if the conditioning is inconsistent).

Then
\m[m_d2e_1]{
&\E[\mac{Z\sim\mu;\\T\sim\U[{[D]}]}]
{\h{\nu_{T,g(Z,1)\dc g(Z,T-1)}^\mu}}\\
&\tbbb=\E[{T\sim\U[{[D]}]}]
{\hh[X\sim\mu]{g(X,T)}{g(X,1)\dc g(X,T-1)}}\\
&\tbbb=\fr1D\tm
\sum_{t=1}^{D}
\hh[X\sim\mu]{g(X,t)}{g(X,1)\dc g(X,t-1)}\\
&\tbbb=\fr1D\tm\h[X\sim\mu]{h(X)}
\ge \fr1D\tm(1-\eps)\tm m = (1-\eps)\tm\fr mD
.}
Accordingly,
\m{
&\E[\mac{Z\sim\mu;\\T\sim\U[{[D]}]}]
{\norm[1]{\nu_{T,h(Z,1)\dc h(Z,T-1)}^\mu-\U[\01^{\dr mD}]}^2}\\
&\tbbb\le 2\ln2 \tm \E[\mac{Z\sim\mu;\\T\sim\U[{[D]}]}]
{\fr mD-\h{\nu_{T,g(Z,1)\dc g(Z,T-1)}^\mu}}\\
&\tbbb\le 2\ln2 \tm \eps \tm \fr mD
,}
where the first inequality is \fctref{c_mu_U} and the second one is~\bref{m_d2e_1}.
The result follows from the concavity of square root.
\prfend

\para{Deterministic \e{entropy} condensers from very strong extractors.}
It is not hard to see that a very strong extractor is not necessarily a good deterministic min-entropy condenser:\ while we require from very strong extractors to behave nearly-uniformly only \e{on average}, bounding the \e{min-entropy} of the condenser's output requires that the probability of a most likely (i.e., \e{worst-case}) output value is not too high.

On the other hand, we have seen that deterministic \e{entropy} condensers are very strong extractors (\clmref{cl_cond_extr}).
It turns out that the connection between \e{very strong extractors} and \e{deterministic entropy} (as opposed to min-entropy) \e{condensers} is two-way:

\clm[cl_extr_cond]{Let $\delta\in[0,1]$, $\Cl F$ be a family of distributions over $\01^n$ and $h:\01^n\times[D]\to\01^m$ be a very strong \f\delta-extractor for $\Cl F$.
Then
\medskip
\[
\h[X\sim\mu]{h(X,1)\dc h(X,D)}
\ge D\tm m - D\tm m\tm\delta - D\tm\sq{4\log e\tm m \tm\delta}
\]
for every $\mu\in\Cl F$.
That is, the \e{entropy rate} of $\l(h(X,s)\r)_{s=1}^{D}$ is at least
$1-\delta - \sq{4\log e\tm\fr{\delta}m}$.
}

\prfstart
Similarly to \defiref{def_extr_strong}, for every $s\in[D]$ and $y_1\dc y_{s-1}\in\01^m$ let $\nu_{s,y_1\dc y_{s-1}}^\mu$ be the distribution of $h(X,S)$ when $X\sim\mu$, conditioned on $[S=s, h(X,1)=y_1\dc h(X,s-1)=y_{s-1} ]$ (let it be undefined if the conditioning is inconsistent).

Then
\m[P]{
\h[X\sim\mu]{h(X,1)\dc h(X,D)}\hspace{-64pt}&\\
&= \sum_{t=1}^{D}
\hh[X\sim\mu]{h(X,t)}{h(X,1)\dc h(X,t-1)}\\
&= D\tm\E[{T\sim\U[{[D]}]}]
{\hh[X\sim\mu]{h(X,T)}{h(X,1)\dc h(X,T-1)}}\\
&= D\tm\E[\mac{Z\sim\mu;\\T\sim\U[{[D]}]}]
{\h{\nu_{T,h(Z,1)\dc h(Z,T-1)}^\mu}}\\
&= Dm - D\tm\E[\mac{Z\sim\mu;\\T\sim\U[{[D]}]}]
{m-\h{\nu_{T,h(Z,1)\dc h(Z,T-1)}^\mu}}  \malabel{m_e2d_1}
.}
By \clmref{c_mu_U},
\m[P]{
&\E[\mac{Z\sim\mu;\\T\sim\U[{[D]}]}]
{m-\h{\nu_{T,h(Z,1)\dc h(Z,T-1)}^\mu}}\\
&\tbb\le \fr m2 \tm\E[\mac{Z\sim\mu;\\T\sim\U[{[D]}]}]
{\norm[1]{\nu_{T,h(Z,1)\dc h(Z,T-1)}^\mu - \U[\01^m]}}\\
&\tbbb + \sq{2\log e\tm m} \tm\E[\mac{Z\sim\mu;\\T\sim\U[{[D]}]}]
{\sq{\norm[1]{\nu_{T,h(Z,1)\dc h(Z,T-1)}^\mu - \U[\01^m]}}}\\
&\tbb\le \fr m2 \tm\E[\mac{Z\sim\mu;\\T\sim\U[{[D]}]}]
{\norm[1]{\nu_{T,h(Z,1)\dc h(Z,T-1)}^\mu - \U[\01^m]}}\\
&\tbbb + \sq{2\log e\tm m \tm\E[\mac{Z\sim\mu;\\T\sim\U[{[D]}]}]
{\norm[1]{\nu_{T,h(Z,1)\dc h(Z,T-1)}^\mu - \U[\01^m]}}}
,}
where the latter inequality follows from the concavity of square root.
As $h$ is a very strong \f\delta-extractor,
\m{
\E[\mac{Z\sim\mu;\\T\sim\U[{[D]}]}]
{\norm[1]{\nu_{T,h(Z,1)\dc h(Z,T-1)}^\mu - \U[\01^m]}} \le 2\delta
}
and
\m{
\E[\mac{Z\sim\mu;\\T\sim\U[{[D]}]}]
{m-\h{\nu_{T,h(Z,1)\dc h(Z,T-1)}^\mu}}
\le m\tm\delta + \sq{4\log e\tm m \tm\delta}
.}

From~\bref{m_e2d_1},
\m{
\h[X\sim\mu]{h(X,1)\dc h(X,D)}
\ge D\tm m - D\tm m\tm\delta - D\tm\sq{4\log e\tm m \tm\delta}
,}
and the result follows.
\prfend

\para{Conclusion.}
We have established the following.

\theo[th_det_cond_v_s_extr]{Let $\Cl F$ be a family of distributions over $\01^n$.

If  $h:\01^n\to\01^m$ is a deterministic entropy condenser for $\Cl F$ of entropy rate at least $1-\eps$, where $\eps\le\fr{2D}{\ln 2\tm m}$ for some $D|m$, then $g_h:\01^n\times[D]\to\01^{\dr mD}$, defined as
\m{
g_h(x,s) = h(x)\upharpoonright_{\fr{(s-1)m}D+1,\dots,\fr{sm}D},
}
is a very strong $\sq{\fr{\ln 2}2\eps\tm\fr mD}$-extractor for $\Cl F$.

If $\delta\in[0,1]$ and $g:\01^n\times[D]\to\01^m$ is a very strong \f\delta-extractor for $\Cl F$, then $h_g:\01^n\to\01^{D\tm m}$, defined as
\m{
h_g(x) = h(X,1)\dcirc h(X,D)
}
is a deterministic entropy condenser for $\Cl F$ of entropy rate at least $1-\delta - \sq{4\log e\tm\fr{\delta}m}$.
}

\section{Conclusions}

We conclude our article with an open problem.
We have shown that there is an essential difference between \SV-sources and strong \SV-sources: for the former, there are no non-trivial condensers, while for the latter we have constructed them.
But this only concerns \emph{min-entropy,} so we pose as an open problem the following.

\begin{problem}
Does there exist non-trivial \emph{entropy} condensers for {\SVd} sources? If so, how much can one condense the entropy of these sources?
\end{problem}

\phantomsection
\addcontentsline{toc}{section}{Acknowledgements}

\sect*{Acknowledgements}

We are grateful to anonymous reviewers for a number of very useful suggestions.

\phantomsection

\end{document}